\documentclass[10pt,aps,showpacs,pra,twocolumn,superscriptaddress]{revtex4-2}

\usepackage[utf8]{inputenc}
\usepackage[T1]{fontenc}
\usepackage[polish,english]{babel}
\usepackage{times}
\usepackage{amssymb}
\usepackage{mathtools}
\usepackage{amsthm}
\usepackage{MnSymbol}
\usepackage{dsfont}
\usepackage{bm}
\usepackage{textcomp}
\usepackage{xcolor}
\usepackage{physics}
\usepackage[shortlabels]{enumitem}
\usepackage{dcolumn}
\usepackage{graphicx}
\usepackage[colorlinks=true,citecolor=blue,urlcolor=blue]{hyperref}

\theoremstyle{definition}
\newtheorem{theorem}{Theorem}
\newtheorem{observation}{Observation}
\newtheorem{definition}{Definition}
\newtheorem{proposition}{Proposition}
\newtheorem{conjecture}{Conjecture}

\begin{document}

\title{Equi-Entropic Maps for Four-Partite Quantum States}

\author{Wojciech Bruzda}
\email{w.bruzda@cft.edu.pl}
\affiliation{Center for Theoretical Physics, Polish Academy of Sciences, Aleja Lotnik\'{o}w 32/46, 02-668 Warsaw, Poland}
\affiliation{Center for Quantum-Enabled Computing, Center for Theoretical Physics, Polish Academy of Sciences, Aleja Lotnik\'{o}w 32/46, 02-668 Warsaw, Poland}
\author{Zahra Raissi}
\email{zahra.raissi@uni-muenster.de }
 \affiliation{Department for Quantum Technology, University of M\"unster, 48149 M\"unster, Germany}

\date{\today}

\begin{abstract}

{Absolutely maximally entangled states represent a highly constrained form of multipartite entanglement and play an important role in quantum information theory. We investigate a weaker form of uniformity of entanglement for four-party systems of local dimension $d>2$ that requires the three balanced bipartitions to have equal but not necessarily maximal linear entropy. We introduce a linear map $\Xi$ that enforces exact equality of entropies under reshuffling and partial transposition. The transformation arises as the asymptotic limit of an iterative averaging procedure and admits a group-theoretic description in terms of permutations of tensor indices. For Haar-random unitary inputs, a leading-moment analysis supported by numerical simulations predicts highly entangled outputs whose common entropy approaches the maximal value as the local dimension grows. We characterize the algebraic structure, fixed points, and asymptotic behavior of this map and its relation to two-unitary matrices and orthogonal Latin squares.}
\end{abstract}

\maketitle

\section{Introduction}

Random pure states in high-dimensional multipartite systems are typically strongly entangled~\cite{Hayden2006}. More generally, the distribution of entanglement across different bipartitions provides an important characterization of multipartite quantum states~\cite{Sc04,Facchi_2010}. Absolutely maximally entangled (AME) states constitute an extremal case, as all their reductions to half of the parties are maximally mixed~\cite{HCLRL12,GALRZ15}. They are closely connected to quantum error-correcting codes, orthogonal arrays, combinatorial designs and multiunitary matrices ~\cite{Helwig_2013,GZ2014,GALRZ15,Goyeneche_2018,Raissi_2018,Raissi_2017,Burchardt_2020}. For a recent comprehensive review of AME states, including their constructions, applications and relation to multiunitary operators, see Ref.~\cite{RMBRLZ26}.

In this work, we focus on four-party systems with equal local dimension $d$. Their three balanced bipartitions admit a common $d^2\times d^2$ matrix representation and are related by reshuffling and partial transposition~\cite{GALRZ15}. This structure gives rise to two-unitary matrices, whose existence is equivalent to the existence of AME$(4,d)$ states. It also suggests a weaker question: can one impose the same entanglement entropy across all three balanced bipartitions without requiring that entropy to be maximal?

Motivated by this question and by earlier work on distributing entanglement as uniformly as possible over inequivalent bipartitions~\cite{Facchi_2010,FFPP08}, we introduce a linear map $\Xi$, constructed from reshuffling and partial transposition, whose nonzero outputs have identical normalized
linear entropies across the three balanced bipartitions.

The map separates entropic uniformity from maximality, i.e., exact AME states satisfy both conditions, while nonzero elements in the image of $\Xi$ are guaranteed to satisfy only the former but not necessarily the latter. To the best of our knowledge, this iterative symmetrization and its parity-dependent equi-entropic limits have not previously been studied in this setting.

We characterize the algebraic structure and fixed points of $\Xi$ using the group action generated by reshuffling and partial transposition. Although generic outputs are not two-unitary, special fixed points correspond to exact AME$(4,d)$ states. For permutation matrices, the fixed-point condition imposes an additional cyclic symmetry on the associated orthogonal Latin squares. Numerical results reveal such symmetric two-unitary fixed points in several dimensions and motivate the conjecture formulated below.

For Haar-random unitary inputs, analytical estimates supported by numerical simulations indicate that the common output entropy approaches its maximal value as the local dimension grows, with an average deficit of order $d^{-2}$. Thus, the symmetry imposed by $\Xi$ enforces exact equality of the three entropies, while high-dimensional random-matrix behavior makes their common value close to maximal. We also relate the entropy deficit to the deviation of the corresponding reduced states from maximal mixedness, however, this relation does not imply geometric closeness to an exact two-unitary matrix.

The remainder of the paper is organized as follows. In Sec.~\ref{sec:AME4d}, we recall four-party AME states, two-unitary matrices, and the normalized linear entropy used throughout the work. In Sec.~\ref{sec:equi-entropic-map}, we introduce and analyze the map $\Xi$, its fixed points, its relation to orthogonal Latin squares, and its large-dimensional behavior. We conclude in Sec.~\ref{sec:summary}.

\section{Absolutely Maximally Entangled States of Four Parties}\label{sec:AME4d}

Absolutely maximally entangled states are highly entangled quantum states which generalize maximally entangled Bell states to larger number of parties and higher local dimension. Such states can be defined in a general scenario of $n$ parties each of any local dimension~\cite{GALRZ15, BZ26}. In this work we restrict attention to the case where the number of parties is limited~to~$4$.

Let $\mathcal{H}=\mathcal{H}_A\otimes\mathcal{H}_B\otimes\mathcal{H}_C\otimes\mathcal{H}_D \cong(\mathbb{C}^d)^{\otimes 4}$ be a composite Hilbert space, where each individual subspace has local dimension $d>2$.

\begin{definition}A pure quantum state $|\psi\rangle\in\mathcal{H}$ is called absolutely maximally entangled, denoted AME$(4, d)$, if every two-party reduced density matrix is maximally mixed, i.e.,
$\rho_Q = \mathrm{Tr}_{\overline Q}|\psi\rangle\langle\psi|=\rho_*=\mathbb{I}_{d^2}/d^2$, where $Q\in\{AB, AC, AD\}$ represents any pair of subsystems and $\overline{Q}$ is its complement.
\end{definition}
This definition generalizes straightforwardly to any (not necessarily even) number of parties, $n\geqslant 2$.

Consider an arbitrary complex-valued matrix $X\in\mathbb{C}^{d^2\times d^2}$ for $d>2$, addressed by a two-index $X_{ij,kl}$,  which describes a bipartite quantum system.

\begin{definition}[Cf. Ref.\cite{GALRZ15}]
    A unitary matrix $U\in\mathbb{U}(d^2)$, with matrix elements $U_{ij,kl}$, is called two-unitary if its reshuffled $U_{ij,kl}^{\mathrm{R}} = U_{ik,jl}$ and partially transposed forms $U_{ij,kl}^{\Gamma} = U_{il,kj}$ are also unitary matrices.
\end{definition}
Note that both operations can be applied to any matrix $X$, provided that the dimensionality allows for these particular rearrangements.

The limitation to four parties allows us to write a compact matrix representation of an AME$(4, d)$ state. Namely, the following observation is true for any $d>2$.
\begin{observation}
$\exists_{\,|\psi\rangle\,\in\,\mathrm{AME}(4,d)} \Longleftrightarrow$ $\exists_{\,\text{two-unitary matrix }\,U\,\in\,\mathbb{U}(d^2)}$.
\end{observation}
\noindent It can be inferred from the formula
\begin{equation}
    |\psi\rangle = \frac{1}{d}\sum_{j=0}^{d-1}\sum_{k=0}^{d-1}|j\rangle_A\otimes|k\rangle_B\otimes U\big(|j\rangle_C\otimes |k\rangle_D\big)\in\mathbb{C}^{d^4}.\label{ms-correspondence}
\end{equation}
This matrix-state correspondence simplifies the construction of AME states.

To describe any non-zero square matrix $X$ of size $d^2$ quantitatively, we consider a normalized linear entropy
\begin{equation}
S(X)=\frac{d^2}{d^2-1}\left(1-\frac{{\rm Tr}(XX^{\dagger}XX^{\dagger})}{{\rm Tr}^2(XX^{\dagger})}\right)\in[0,1].\label{SL_entropy}
\end{equation}
Note that for $\rho_X = XX^{\dagger}/\mathrm{Tr}(XX^{\dagger})$, which is a well-defined density matrix for any $X\neq 0$, Eq.~\eqref{SL_entropy} simplifies to
\begin{equation}
S(X) = \frac{d^2}{d^2-1}\left(1-\mathrm{Tr}\rho_X^2\right),
\end{equation}
which can be understood as the normalized linear entropy of the state $\rho_X$~\cite{MW2002,Sc04, PhysRevResearch.2.043126}.

Extreme values of $S$ are achieved by special matrices. Obviously, the entropy vanishes for any rank-one matrix and it saturates at $1$, if $XX^{\dagger} = \mathrm{Tr}(XX^{\dagger})\mathbb{I}/d^2$, which means that $X$ is proportional to a unitary matrix. This implies the scaling property of the entropy, $S(kX) = S(X)$ for any non-zero scalar $k\in\mathbb{C}$. Hence, $S$ measures the flatness of the singular values of $X$, i.e., the more uniform the singular values are, the closer $X$ is to a unitary matrix (after a possible normalization).

To examine the three bipartitions that characterize AME$(4, d)$ states, we define the vector of entropies $\big(S(X),S(X^{\mathrm{R}}),S(X^\Gamma)\big)$. If all three components are equal, we denote their common value by $s\equiv s(X) = S(X)=S(X^{\mathrm{R}})=S(X^\Gamma)$. In particular, the condition $s=1$ means that $X$, $X^{\mathrm{R}}$ and $X^\Gamma$ are proportional to unitary matrices, so after adjusting the normalization,
$\mathrm{Tr}(XX^\dagger)=d^2$, this is equivalent to the fact that $X$ is a two-unitary matrix. Hence
\begin{equation}
|\psi\rangle= \frac{|X\rangle\!\rangle}{\sqrt{\mathrm{Tr}(XX^\dagger)}}\in\mathrm{AME}(4, d),
\end{equation}
where $|X\rangle\!\rangle=\mathrm{vec}(X)$ is matrix vectorization.

\section{Construction of Equi-Entropic Map}\label{sec:equi-entropic-map}

\subsection{\texorpdfstring{Algebraic Construction}{}}

Let $0\neq X_0\in\mathbb{C}^{d^2\times d^2}$ for $d\geqslant 2$ be any matrix.
Define a recursive sequence
\begin{equation}
    X_{k+1}\equiv\frac{1}{2}X_k^{\mathrm{R}}+\frac{1}{2}X_k^{\rm\Gamma} \ : \ k\in\mathbb{N}\cup\{0\}.\label{Xk-definition}
\end{equation}
Examination of the few initial elements of the sequence $(X_k)$ provides the closed form for the iterative formula~\eqref{Xk-definition}. One has
\begin{align}
    X_1&=\frac{1}{2}\left(X_0^{\mathrm{R}}+X_0^{\rm\Gamma}\right)\\
    X_2&=\frac{1}{4}\left(X_0^{\mathrm{RR}}+X_0^{\mathrm{R}\Gamma}+X_0^{\Gamma\mathrm{R}}+X_0^{\Gamma\Gamma}\right)\\
    X_3&=\frac{1}{8}\left(X_0^{\mathrm{RRR}}+X_0^{\mathrm{R}\Gamma \mathrm{R}}+X_0^{\Gamma\mathrm{RR}}+...+X_0^{\Gamma\Gamma\Gamma}\right)\\
    &\vdots\nonumber
\end{align}
In general, the $k^{\rm th}$ element takes the form
\begin{equation}
    X_k = \frac{1}{2^k}\sum_{g_j\in\{{\mathrm{R}, \Gamma}\}}X_0^{g_1g_2\cdot\cdot\cdot g_k},\label{Xk}
\end{equation}
where the summation is performed over all $2^k$ strings
$(g_1,\ldots,g_k)$ with elements in $\{{\mathrm{R}},\Gamma\}$. It is understood that operations are performed in the following order $X^{\mathrm{R}\Gamma}=(X^{\mathrm{R}})^\Gamma$.

To simplify (and slightly abusing) notation, we rewrite reshuffling and partial transpose as operators acting on a matrix, ${\mathrm{R}}(X) \equiv X^{\mathrm{R}}$ and ${\rm\Gamma}(X) \equiv X^{\rm\Gamma}$. Hence $X^{g_1\cdot\cdot\cdot g_k}=g_k\circ\cdot\cdot\cdot\circ g_1(X)=\prod_{j}g_j(X)$ denotes standard composition of operators. Now Eq.~\eqref{Xk} can be expressed as
\begin{equation}
    X_k=\frac{1}{2^k}\left(\sum_{g_j\in\{\mathrm{R},\Gamma\}} \prod_{j=k}^1 g_j\right)(X_0).
\end{equation}

Both maps ${\mathrm{R}}$ and $\Gamma$ are involutions, i.e., ${\mathrm{R}}^2=\Gamma^2=\mathbb{I}$. They do not commute, however when applied sequentially to a given matrix they always reduce to one form from the set
\begin{equation}
    \mathbb{G} = \Big\{\mathbb{I},{\mathrm{R}},\Gamma, {\mathrm{R}}\Gamma, \Gamma{\mathrm{R}}, {\mathrm{R}}\Gamma{\mathrm{R}}=\Gamma{\mathrm{R}}\Gamma\Big\},
\end{equation}
where $\mathbb{I}$ is the identity operator. This fact is a consequence of the definition of ${\mathrm{R}}$ and ${\rm\Gamma}$ and one can infer the structure of $\mathbb{G}$ by examining the cyclic rearrangements of multi-indices. In other words, any sequence of operators ${\mathrm{R}}$ and ${\rm\Gamma}$ reproduces one particular element from $\mathbb{G}$, so they generate a group with the binary operation, $\circ$, being a composition of two operations. The group table reads:
\begin{equation}
\begin{array}{c!{\vrule width 1.5pt}c|c|c|c|c|c}
\circ & \mathbb{I} & {\mathrm{R}} & \Gamma & {\mathrm{R}}\Gamma & \Gamma{\mathrm{R}} & {\mathrm{R}}\Gamma{\mathrm{R}}\\
\noalign{\hrule height 1.5pt}
\mathbb{I} & \mathbb{I} & {\mathrm{R}} & \Gamma & {\mathrm{R}}\Gamma & \Gamma{\mathrm{R}} & {\mathrm{R}}\Gamma{\mathrm{R}}\\
\hline
{\mathrm{R}} & {\mathrm{R}} & \mathbb{I} & {\mathrm{R}}\Gamma & \Gamma & {\mathrm{R}}\Gamma{\mathrm{R}} & \Gamma {\mathrm{R}}\\
\hline
\Gamma & \Gamma & \Gamma{\mathrm{R}} & \mathbb{I} & \Gamma{\mathrm{R}}\Gamma & {\mathrm{R}} & {\mathrm{R}}\Gamma\\
\hline
{\mathrm{R}}\Gamma & {\mathrm{R}}\Gamma & {\mathrm{R}}\Gamma{\mathrm{R}} & {\mathrm{R}} & \Gamma{\mathrm{R}} & \mathbb{I} & \Gamma\\
\hline
\Gamma{\mathrm{R}} & \Gamma{\mathrm{R}} & \Gamma & \Gamma{\mathrm{R}}\Gamma & \mathbb{I} & {\mathrm{R}}\Gamma & {\mathrm{R}}\\
\hline
{\mathrm{R}}\Gamma{\mathrm{R}} & {\mathrm{R}}\Gamma{\mathrm{R}} & {\mathrm{R}}\Gamma & \Gamma{\mathrm{R}} & {\mathrm{R}} & \Gamma & \mathbb{I}
\end{array}.\label{group-table}
\end{equation}

\begin{proposition}\label{Xk-convergence}
    For every $d\geqslant 2$ and every matrix $X_0\in\mathbb C^{d^2\times d^2}$, the sequence defined in~\eqref{Xk-definition} contains two subsequences for even and odd indices. Both of them converge.
\end{proposition}
\begin{proof}
Define $\xi\equiv\frac12({\mathrm{R}}+\Gamma)$.
Then the sequence~\eqref{Xk-definition} can be written as $X_k=\xi^k(X_0)$.  Obviously
\begin{align}
\xi^2 &= \frac14 \left( {\mathrm{R}}^2 + {\mathrm{R}}\Gamma + \Gamma{\mathrm{R}} + \Gamma^2 \right)
= \frac12\mathbb{I} + \frac14{\mathrm{R}}\Gamma + \frac14\Gamma{\mathrm{R}}\\
&=\frac12\mathbb{I}+\frac14 \chi+\frac14 \chi^2,
\end{align}
where $\chi=\mathrm{R}\Gamma$. Because $\chi^3=\mathbb{I}$, then $\Pi\equiv(\mathbb{I}+\chi+\chi^2)/3 = \Pi^2$ is a projection and $\xi^2=\Pi+(\mathbb{I}-\Pi)/4$. This allows us to write the even powers of $\xi$ as
\begin{equation}
\xi^{2k} = (\xi^2)^k = \Pi + (\mathbb{I}-\Pi)/4^k\longrightarrow\Pi\,(k\to\infty).
\end{equation}
Hence $X_{2k}=\xi^{2k}(X_0)$ and
\begin{equation}
    \exists\,\Xi_{\texttt{e}}(X_0)\equiv\lim_{k\to\infty} X_{2k} = \frac{1}{3}\left(X_0+X_0^{{\mathrm{R}}\Gamma}+X_0^{\Gamma{\mathrm{R}}}\right).
\end{equation}

As for the odd powers, one has the following decomposition
$X_{2k+1}=\xi^{2k+1}(X_0)=\xi(\xi^{2k}(X_0))$. Then 
\begin{align}
\xi^{2k+1} \stackrel{k\to\infty}{\longrightarrow} \xi \Pi&=\frac{1}{2}(\mathrm{R}+\Gamma)\frac{1}{3}(\mathbb{I}+\mathrm{R}\Gamma+\Gamma\mathrm{R})\\
&=\frac{1}{3}(\mathrm{R}+\Gamma+\mathrm{R}\Gamma\mathrm{R}).
\end{align}
Hence $X_{2k+1}=\xi^{2k+1}(X_0)$ and
\begin{equation}
\exists\,\Xi_{\texttt{o}}(X_0) \equiv \lim_{k\to\infty}X_{2k+1} = \frac13\left(X_0^{\mathrm{R}}+X_0^\Gamma+X_0^{{\mathrm{R}}\Gamma{\mathrm{R}}}\right).
\end{equation}
In conclusion, both subsequences converge.
\end{proof}

The proof of Proposition~\ref{Xk-convergence} provides two compact formulas:
\begin{align}
\Xi_{\texttt{e}}(X) &= \frac{1}{3}\left(X + X^{\mathrm{R}\Gamma}+X^{\rm\Gamma R}\right),\label{eq:Xi1}\\
\Xi_{\texttt{o}}(X) &= \frac{1}{3}\left(X^{\mathrm{R}} + X^{\Gamma}+X^{\mathrm{R}\Gamma\mathrm{R}}\right).\label{eq:Xi2}
\end{align}
The most important feature of both maps $\Xi_{\texttt{e}}$ and $\Xi_{\texttt{o}}$ is that their images are equi-entropic with respect to the three rearrangements $Y$, $Y^{\mathrm{R}}$ and $Y^\Gamma$. This property is formally expressed in the following theorem.

\begin{theorem}
Let $X\in\mathbb{C}^{d^2\times d^2}$ for $d\geqslant 2$, and let $Y=\Xi_{\texttt{e}}(X)$ or $Y=\Xi_{\texttt{o}}(X)$. If $Y\neq 0$ then
$S(Y)=S(Y^{\mathrm{R}})=S(Y^\Gamma)$.
\end{theorem}
\begin{proof}
We prove the claim for $Y=\Xi_{\texttt{e}}(X)$ because the proof for $Y=\Xi_{\texttt{o}}(X)$ is almost identical. By definition, $Y=\left(X+X^{{\mathrm{R}}\Gamma}+X^{\Gamma{\mathrm{R}}}\right)/3$. Hence $Y$ is invariant under the even permutations generated by ${\mathrm{R}}\Gamma$ and $\Gamma{\mathrm{R}}$, namely
$Y^{{\mathrm{R}}\Gamma}=Y$ and $Y^{\Gamma{\mathrm{R}}}=Y$.
Moreover,
\begin{equation}
Y^{\mathrm{R}}=Y^\Gamma=\frac13\left(X^{\mathrm{R}}+X^\Gamma+X^{{\mathrm{R}}\Gamma{\mathrm{R}}}\right)=\Xi_{\texttt{o}}(X).
\end{equation}
We now compare the Gram matrices,
\begin{align}
\left(Y^{\mathrm{R}}(Y^{\mathrm{R}})^\dagger\right)_{ij,mn}&=\sum_{k,l}Y^{\mathrm{R}}_{ij,kl}\,\overline{Y^{\mathrm{R}}_{mn,kl}}=\sum_{k,l}Y_{ik,jl}\,\overline{Y_{mk,nl}}\\
&=\sum_{k,l}Y_{ij,lk}\,\overline{Y_{mn,lk}}\label{invariant-substitution}\\
&=\left(YY^\dagger\right)_{ij,mn},
\end{align}
where in Eq.~\eqref{invariant-substitution} we used the invariance
$Y^{{\mathrm{R}}\Gamma}=Y$ which is equivalent to $Y_{il,jk}=Y_{ij,kl}$, so effectively, it renames indices. Similarly, $Y^{\Gamma{\mathrm{R}}}=Y$ implies
$Y^\Gamma(Y^\Gamma)^\dagger=YY^\dagger$. Hence $YY^\dagger=Y^{\mathrm{R}}(Y^{\mathrm{R}})^\dagger=Y^\Gamma(Y^\Gamma)^\dagger$.
Since $S(X)$ depends only on $XX^\dagger$, it follows that
$S(Y)=S(Y^{\mathrm{R}})=S(Y^\Gamma)$.
\end{proof}

In fact, the proof establishes the stronger identity
\begin{equation}
YY^\dagger=Y^{\mathrm{R}}(Y^{\mathrm{R}})^\dagger= Y^\Gamma(Y^\Gamma)^\dagger,
\end{equation}
so the three balanced reduced states have identical spectra under the natural identifications of the corresponding Hilbert spaces. Related recent work has studied constrained subspaces in which several bipartitions have identical entanglement spectra, under the name ``bundling'' of bipartite entanglement~\cite{DSDL26}. In the present construction the relevant subspace is obtained explicitly from reshuffling and partial transposition and is tailored to the three balanced cuts of a four-party tensor.

The convergence of the two subsequences is exponentially fast. In numerical simulations, after several dozen iterations the even and odd subsequences are typically very close to their respective limits $\Xi_{\texttt{e}}(X_0)$ and $\Xi_{\texttt{o}}(X_0)$. However, in general, the full sequence $(X_k)$ need not converge and it happens if and only if $\Xi_{\texttt{o}}(X_0)=\Xi_{\texttt{e}}(X_0)$.

The two limiting maps admit a natural interpretation in terms of the group $\mathbb{G}$. This group decomposes into the cyclic subgroup $\mathbb{G}_{\texttt{e}} = \{\mathbb{I},\mathrm{R}\Gamma,\Gamma\mathrm{R}\}$ and its complementary coset $\mathbb{G}_{\texttt{o}} = \{\mathrm{R},\Gamma,\mathrm{R}\Gamma\mathrm{R}\}$. Accordingly,
\begin{equation}
\Xi_{\texttt{e}}(X) = \frac{1}{3} \sum_{g\in\mathbb{G}_{\texttt{e}}}X^g\quad\text{and}\quad \Xi_{\texttt{o}}(X) = \frac{1}{3} \sum_{g\in\mathbb{G}_{\texttt{o}}}X^g. \end{equation}
This decomposition also explains why the original sequence does not converge in general. Each application of either generator $\mathrm{R}$ or $\Gamma$ changes the parity of the corresponding group element. Consequently, the even iterates remain supported on $\mathbb{G}_{\texttt{e}}$ while the odd ones remain on $\mathbb{G}_{\texttt{o}}$. Within each parity class, the iterative averaging becomes asymptotically uniform, yielding the two limiting branches above. Averaging the two branches gives the full group-averaging projector \begin{equation}
\Xi_{\texttt{full}}(X) \equiv \frac{1}{2}\big( \Xi_{\texttt{e}}(X)+\Xi_{\texttt{o}}(X) \big) = \frac{1}{6}\sum_{g\in\mathbb{G}}X^g. \label{eq:Xi-full-average} 
\end{equation}
Equivalently,
\begin{equation}
\Xi_{\mathrm{full}}(X) = \lim_{k\to\infty} \frac{X_{2k}+X_{2k+1}}{2}. 
\end{equation}
The map $\Xi_{\mathrm{full}}$ projects onto the subspace invariant under the full group $\mathbb{G}$, while $\Xi_{\texttt{e}}$ projects only onto the subspace invariant under its cyclic subgroup $\mathbb{G}_{\texttt{e}}$. Hence, the full group average is not equivalent to either limiting branch because it imposes the stronger conditions $Y=Y^{\mathrm{R}}=Y^\Gamma$ removing the distinction between the even and odd limits.

This distinction is important because the weaker cyclic symmetry of $\Xi_{\texttt{e}}$ can admit two-unitary fixed points while invariance under the full group is incompatible with unitarity, see Subsection~\ref{Xi-properties}. Related symmetry-constrained constructions of multiunitary tensors and multipartite states with geometrically selected maximally entangled cuts have been studied in Ref.~\cite{MPW2024}.

\medskip

Let $\Xi$ denote either of the two limiting maps
$\Xi_{\texttt e}$ and $\Xi_{\texttt o}$. In the following, we will specify the parity of $\Xi$ only when it is necessary.

\subsection{\texorpdfstring{Properties of the Map $\Xi$}{}}\label{Xi-properties}

The limiting map $\Xi_{\texttt{e}}$ is a projector. Indeed, one can readily check that $\Xi_{\texttt{e}}^2=\Xi_{\texttt{e}}$. However, it is not the case for the second branch. Instead, $\Xi_{\texttt{o}}^2=\Xi_{\texttt{e}}$ and the anti-commutator $\{\Xi_{\texttt{e}},\Xi_{\texttt{o}}\}=\Xi_{\texttt{e}}\Xi_{\texttt{o}}+\Xi_{\texttt{o}}\Xi_{\texttt{e}}=2\Xi_{\texttt{o}}$. This implies that $\Xi_{\texttt{o}}^3=\Xi_{\texttt{o}}$. This makes $\Xi_{\texttt{o}}$ an involution on the image of $\Xi_{\texttt{e}}$, that is, if $Y\in\textrm{Image}(\Xi_{\texttt{e}})$ then $\Xi_{\texttt{e}}(Y)=Y$ and $\Xi_{\texttt{o}}^2(Y)=\Xi_{\texttt{e}}(Y)=Y$.

This raises two natural questions. First, for what matrices $X$, $\Xi(X)=X$. Second, given a generic $X$, what is the average entropy $S(\Xi(X))$. In the following we will address these two problems.

\subsubsection{\texorpdfstring{Fixed Points}{}}

The fixed points of the two limiting maps have a simple description based on the group theory. Adopting the notation used in the proof of Proposition~\ref{Xk-convergence}, we see that $\Xi_{\texttt{e}}$ is the averaging projection onto the subspace invariant under the cyclic subgroup generated by $\chi=\mathrm{R}\Gamma$, see~\eqref{group-table}. Hence
$\Xi_{\texttt{e}}(X)=X$ if and only if $X^{{\mathrm{R}}\Gamma}=X^{\Gamma{\mathrm{R}}}=X$.

Any fixed point of $\Xi_{\texttt{o}}$ must also be a fixed point of $\Xi_{\texttt{e}}$, because of $\Xi_{\texttt{o}}^2=\Xi_{\texttt{e}}$. Hence $\Xi_{\texttt{o}}(X)=X$ if and only if $X$ is invariant under all elements of $\mathbb{G}$, equivalently
\begin{equation}
X=X^{\mathrm R}=X^\Gamma=X^{{\mathrm R}\Gamma}=X^{\Gamma{\mathrm R}}
=X^{{\mathrm R}\Gamma{\mathrm R}}.
\end{equation}

It is known that there are no unitary matrices $X$ such that $X=X^{\mathrm{R}}=X^{\rm\Gamma}$~\cite{BRZ24}.
If there was such a strongly two-unitary matrix $X$, then
\begin{align}
\Xi_{\texttt{e}}(X)&=\frac{1}{3}X+\frac{1}{3}X^{\mathrm{R}\Gamma}+\frac{1}{3}X^{\rm\Gamma R}\\
&=\frac{1}{3}X+\frac{1}{3}X+\frac{1}{3}X=\Xi_{\texttt{o}}(X)=X.
\end{align}
One can relax the condition and require only $X^{\mathrm{R}}=X^{\rm \Gamma}$. This automatically guarantees that $X=\Xi_{\texttt{e}}(X)$. In particular, every nonzero constant matrix, $\forall_{i,j,k,l}:X_{ij,kl} = \texttt{constant}$, is a fixed point for both branches of $\Xi$, however, their entropy vanishes, $s=0$. Nevertheless, for some dimensions $d$ one can find a matrix $X$ which is a two-unitary fixed point of either branch of $\Xi$.

Consider $d=3$ and the permutation matrix
\begin{equation}
    P_9= [1,   9,   5,   8,   4,   3,   6,   2,   7],
\end{equation}
 where each integer denotes the position of unity in the consecutive column of $P_{9}$. This matrix is two-unitary, $s(P_9)=1$. Matrix $X_9(\varphi)$ defined as
\begin{equation}
X_9(\varphi)=e^{i\varphi}  P_{9}\left[\begin{array}{ccc}
1 & 0 & 0\\
0 & 0 & 1 \\
0 & 1 & 0
\end{array}
\right]\otimes \mathbb{I}_3,
\end{equation}
with $\varphi\in[0, 2\pi)$, represents a family of fixed points, $\Xi_{\texttt{e}}(X_9(\varphi))=X_9(\varphi)$. Similarly, for $d=4$ there exists a special two-unitary permutation matrix
\begin{equation}
    P_{16}=[ 1,16,6,11,15,2,12,5,8,9,3,14,10,7,13,4].
\end{equation}
Matrix $X_{16}(\varphi)= e^{i\varphi}P_{16}\big(Q_C \otimes Q_D\big)$, with 
\begin{equation}
Q_C= \left[\begin{array}{cccc}
0 & 0 & 0 & 1\\
1 & 0 & 0 & 0\\
0 & 0 & 1 & 0\\
0 & 1 & 0 & 0
\end{array}\right], \quad
Q_D = \left[\begin{array}{cccc}
0 & 1 & 0 & 0\\
0 & 0 & 0 & 1\\
0 & 0 & 1 & 0\\
1 & 0 & 0 & 0
\end{array}\right]
\end{equation}
also represents a family of fixed points for $\Xi_{\texttt{e}}$. In both cases, $s(X_j(\varphi))=1$ for every $\varphi\in[0,2\pi)$ and $j\in\{9,16\}$. The right-multiplying local permutation matrices preserve two-unitarity and are chosen so that $X_j(\varphi)$ satisfies the fixed-point condition.

\subsubsection{Relation to Latin Squares}

Two-unitary permutation matrices are naturally related to pairs of orthogonal Latin squares of order $d$~\footnote{Two Latin squares of order $d$ are said to be orthogonal, written OLS($d$), if, when superimposed, they form an array of ordered pairs with no repetitions.}, denoted by OLS($d$). In this view, the additional condition $P=\Xi_{\texttt e}(P)$ imposes a cyclic symmetry on the underlying combinatorial structure. Thus, the problem of finding a two-unitary permutation matrix being a fixed point of $\Xi_{\texttt e}$ can be regarded as a symmetry-constrained version of the problem of constructing OLS. The unconstrained existence problem is solved and a pair of OLS of order $d$ exists for every $d\neq 2,6$~\cite{BSP60,ColbournDinitz2007}. Hence the fixed-point and also centrosymmetry requirements considered here are
genuine additional constraints. For $d=2$ a pair of OLS does not exist, and consequently there is no two-unitary permutation matrix. For four-qubit systems, the best achievable value of the corresponding entropy is $s=8/9$, in agreement with the obstruction to the existence of AME$(4,2)$ states~\cite{HIGUCHI2000213}. The case $d=6$ is Euler's classical ``$36$ officers'' problem and no pair of OLS$(6)$ exists~\cite{Stinson1984}, although its quantum
generalization admits a two-unitary solution~\cite{Euler-PRL, Suhail2024, BZ24, GG25}.

For general $d$ the numerical evidence suggests the following conjecture.
\begin{conjecture}\label{LS-conjecture}
    If $d\not\equiv 2\;(\bmod\;3)$ and $d\neq 6$, then there exists a centrosymmetric~\footnote{A matrix $P$ of size $d^2\times d^2$ is called centrosymmetric if $P_{ij}=P_{d^2+1-i,d^2+1-j}$} two-unitary permutation matrix $P$ of size $d^2$ such that $P = \Xi_{\texttt{e}}(P)$.
\end{conjecture}
It is supported by findings for $d\in\{3,4,7,9,12,13,15,16\}$. The condition $d\not\equiv 2\;(\bmod\;3)$ is necessary. Otherwise, no such permutation can exist. This follows from the cyclic character of the transformation of the indices  $j$, $k$, $l$ induced by the condition $P=P^{\mathrm{R}\Gamma}$ (note that the first index is not affected), and analysis of invariants of this transformation.

Central symmetry (CS) is not strictly required -- see the matrices $X_j$ considered in the previous paragraph. It only provides an additional constraint which significantly reduces the search space while still producing two-unitary in all tested dimensions. Additionally, one can speed-up numerical search for these matrices exploiting the fact that they must fulfill the condition $P^{\mathrm{R}}=P^{\Gamma}$. Examples of such CS matrices for $d=3$, $4$ and $7$ are depicted in Fig.~\ref{fig:PMCS_d347}.
\begin{figure}[ht!]
  \centering
  \includegraphics[width=3.25in]{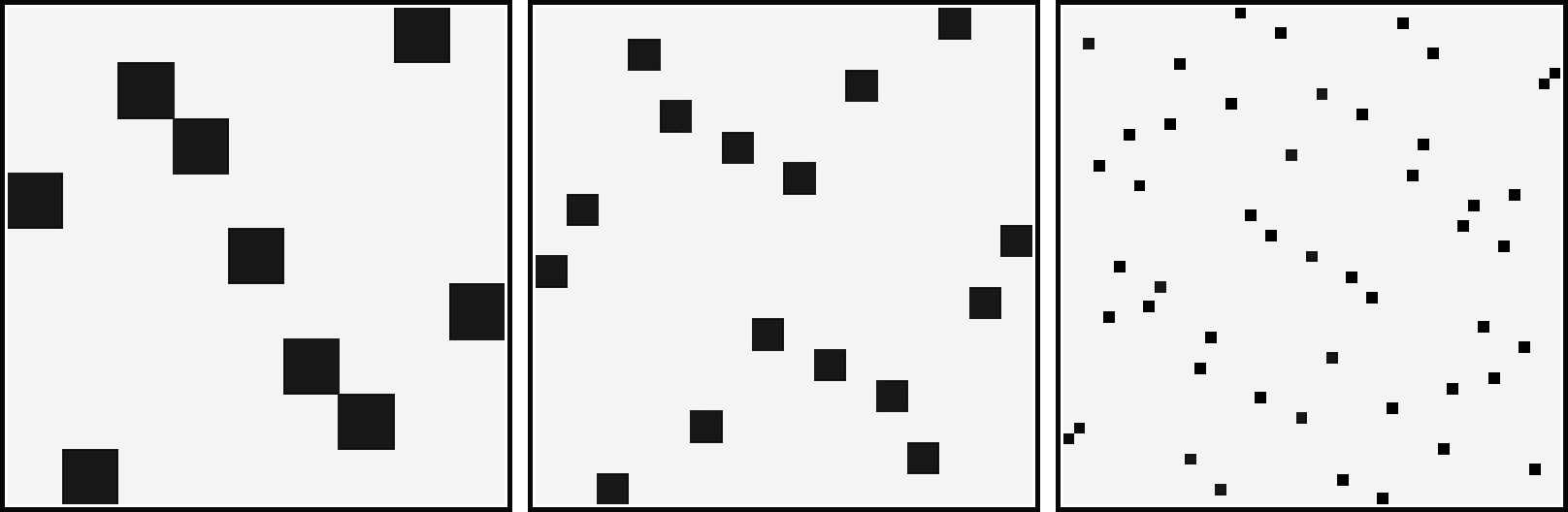}
  \caption{Three examples of two-unitary centrosymmetric permutation matrices $P$ of size $d^2$ for $d\in\{3,4,7\}$ which are fixed points of the even branch $\Xi_{\texttt{e}}$. Black squares indicate unities.}
  \label{fig:PMCS_d347}
\end{figure}

For $d=7$ and $d=9$, our symmetry-constrained search yields complete sets of $d-1$ MOLS. The additional feature is their realization by centrosymmetric permutation fixed points of $\Xi_{\texttt e}$. Explicit arrays are shown in Appendix~\ref{app:LS}. For $d=12$, the numerical search produces a large number of centrosymmetric permutations fixed by $\Xi_{\texttt e}$, but so far no configuration producing more than five MOLS has been found. The current best-known lower bound is also five, while it remains open whether six MOLS of order $12$ exist~\cite{MAVF24}.

\subsubsection{\texorpdfstring{Range of Image and Asymptotics}{}}

Given an initial point $X$ it is unusual that $S(\Xi(X))$ saturates at $1$,
unless $X$ has a very special structure. Nevertheless, a generic choice of the starting point provides matrices $\Xi(X)$ for which its entropy is quite close to unity and, moreover, as the local dimension $d$ grows, $S(\Xi(X))$ approaches $1$ relatively fast.

\begin{proposition}
For a Haar-random unitary matrix $X\in\mathbb{U}(d^2)$, the leading moment approximation predicts
\begin{equation}
1-\mathbb{E}\left[ S\left(\Xi_{\texttt{e}}(X)\right) \right] \sim \frac{1}{d^2} \quad\text{as}\quad d\to\infty.
\end{equation}
\end{proposition}
\begin{proof}
We provide a moment-based derivation of the leading behavior for large $d$.

Let $X\in\mathbb{U}(d^2)$ be a Haar-random unitary matrix and let
$Y=\Xi_{\texttt{e}}(X)$. We will estimate the expectation value $\mathbb{E}\left[S(Y)\right]$. We start with  $\mathrm{Tr}(YY^\dagger)$, it expands as
\begin{align}
\mathrm{Tr}\left(YY^\dagger\right)
&=\frac19\mathrm{Tr}\left[
\big(X+X^{\mathrm{R}\Gamma}+X^{\Gamma\mathrm{R}}\big)
\big(...\big)^\dagger\right] \\
&=\frac19\sum_{g_1,\,g_2\,\in\,\{\mathbb{I},\,\mathrm{R}\Gamma,\,\Gamma\mathrm{R}\}}
\mathrm{Tr}\left(X^{g_1}(X^{g_2})^\dagger\right)\\
&=\frac19\sum_{g\,\in\,\{\mathbb{I},\,\mathrm{R}\Gamma,\,\Gamma\mathrm{R}\}}
\mathrm{Tr}\left(X^g(X^g)^\dagger\right)\label{double-g-sum}\\
&+\frac19
\sum_{g_1\neq g_2}\mathrm{Tr}\left(X^{g_1}(X^{g_2})^\dagger\right).
\end{align}
The first sum in the last equality (Eq.~\eqref{double-g-sum}) contains three terms. Since
each $g$ only permutes the entries of $X$, it preserves the
Frobenius norm. Hence
$\mathrm{Tr}\left(X^g(X^g)^\dagger\right)=\mathrm{Tr}\left(XX^\dagger\right)=d^2$.
For the second sum, we use index notation and, for simplicity, we define $I\equiv ij,kl$. Since each
$g_\alpha$ with $\alpha\in\{1,2\}$, only rearranges the four-index
configuration, we can write that  $X^{g_\alpha}_{I}=X_{g_\alpha(I)}$.
Then
\begin{align}
\mathbb{E}\left[\mathrm{Tr}\left(X^{g_1}(X^{g_2})^\dagger\right)\right]
&=\sum_{I}\mathbb{E}\left[X^{g_1}_{I}\overline{X^{g_2}_{I}}\right]\\
&=\frac{1}{d^2}\#\big\{I:g_1(I)=g_2(I)\big\}.
\end{align}
For $g_1\neq g_2$ from $\{\mathbb{I},\,\mathrm{R}\Gamma,\,\Gamma\mathrm{R}\}$,
the relative permutation $g_2^{-1}g_1$ is a nontrivial cyclic
permutation of only the three indices $j$, $k$, $l$. Its fixed points satisfy $j=k=l$, thus there are $d$ choices for $i$ (which stays intact) and $d$ choices for the common
value of $j=k=l$, which gives $d^2$ terms in total. Finally
\begin{equation}
\mathbb{E}\left[\mathrm{Tr}\left(X^{g_1}(X^{g_2})^\dagger\right)\right]=\frac{d^2}{d^2}=1
\end{equation}
and, for $d\gg 1$,
\begin{equation}
\mathbb{E}\left[\mathrm{Tr}\left(YY^\dagger\right)\right]=\frac19\left(3d^2+6\right) \sim \frac{d^2}{3}.
\end{equation}
A fourth-moment calculation based on Haar-unitary matrix elements~\cite{CS2006} gives, to leading order, 
\begin{equation}
\mathbb{E}\big[ \mathrm{Tr}\big((YY^\dagger)^2\big)\big] \sim \frac{2d^2}{9}. 
\end{equation}
The complete fourth-moment expansion is straightforward but lengthy, hence we retain only the contribution relevant to the leading behavior for $d\to\infty$. We replace the expectation of the ratio entering the entropy by the ratio of its leading moments:
\begin{align}
\mathbb{E}\left[\frac{\mathrm{Tr}\big((YY^\dagger)^2\big)} {\mathrm{Tr}^2(YY^\dagger)}\right] &\simeq \frac{\mathbb{E}\left[\mathrm{Tr}\big((YY^\dagger)^2\big)\right]}{\big( \mathbb{E}\left[\mathrm{Tr}(YY^\dagger)\right] \big)^2 } \\
&\sim \frac{2d^2/9}{(d^2+2)^2/9} \sim \frac{2}{d^2}.
\end{align}
Consequently, 
\begin{align} \mathbb{E}\left[S(Y)\right] &\simeq \frac{d^2}{d^2-1} \left( 1-\frac{2d^2}{(d^2+2)^2} \right) \\
&= 1-\frac{1}{d^2}+O(d^{-4}).
\end{align}
An analogous calculation applies to the odd branch of $\Xi$.
\end{proof}

The numerical results presented in Fig.~\ref{fig:Xi_convergence} support this approximation.

\begin{figure}[ht!]
  \centering
  \includegraphics[width=3.0in]{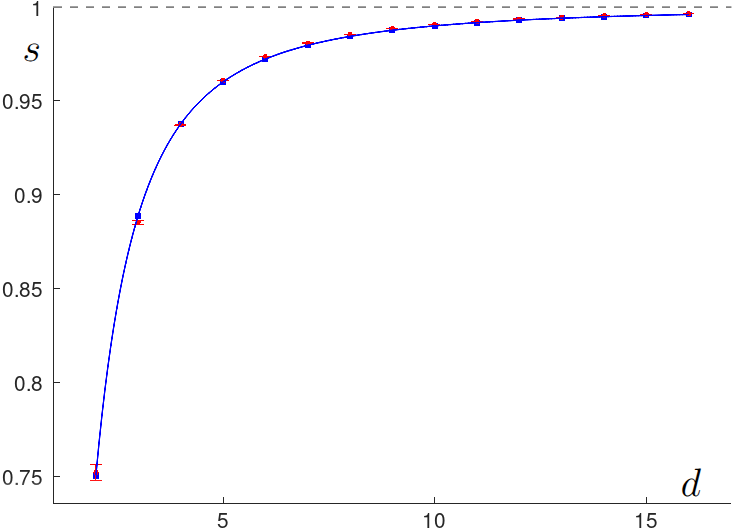}
  \caption{As the local dimension $d$ tends to infinity, $\mathbb{E}\left[S(\Xi_{\texttt{e}}(X))\right]$, for a Haar-random unitary $X\in\mathbb{U}(d^2)$, saturates at $1$. The vertical axis indicates the equal-entropy values $s=S(\Xi_{\texttt{e}}(X))$ defined in Sec.~\ref{sec:AME4d}. Squares on the solid line represent theoretical predictions, while dots with error bars (corresponding to the standard error) show the sample means over $256$ independent Haar-random unitaries.}
  \label{fig:Xi_convergence}
\end{figure}

\subsection{Distance from Maximal Mixedness}

We conclude this section with a remark on the relation between the linear-entropy deficit and the distance from maximal mixedness. The map $\Xi$ enforces equality of the entropies associated with the three balanced bipartitions, but the entropy is not a metric on the space of states. To quantify the deviation of an individual reduced state from maximal mixedness, recall $\rho_X=XX^\dagger/\mathrm{Tr}(XX^\dagger)$ and  $\rho_*=\mathbb{I}_{d^2}/d^2$, and define the Hilbert--Schmidt defect
\begin{equation}
\delta_{\mathrm{HS}}(X) = \left\|\rho_X-\rho_*\right\|_2. \end{equation}
A direct calculation gives
\begin{align}
\delta_{\mathrm{HS}}^2(X) &= \mathrm{Tr}\rho_X^2-\frac{1}{d^2}\\
&= \frac{d^2-1}{d^2}\bigl(1-S(X)\bigr).
\label{eq:HS-entropy}
\end{align}
Thus, the normalized linear-entropy deficit is exactly proportional to the squared HS distance from the maximally mixed state. For a four-party tensor, it is useful to introduce the marginal AME defect
\begin{equation}
\Delta_{\mathrm{AME}}(X) = \max_{g\in\{\mathbb{I},\mathrm{R},\Gamma\}} \left\|\rho_{X^g}-\rho_*\right\|_2.
\end{equation}

The quantity $\Delta_{\mathrm{AME}}$ measures violation of the AME marginal conditions but it is not a metric distance to the set of exact AME states. This quantity vanishes precisely when the normalized vectorization of $X$ is an AME$(4,d)$ state. If $Y=\Xi(X)\neq0$ and $s=S(Y)=S(Y^{\mathrm{R}})=S(Y^\Gamma)$, then Eq.~\eqref{eq:HS-entropy} implies
\begin{align}
\left\|\rho_Y-\rho_*\right\|_2 &= \left\|\rho_{Y^{\mathrm{R}}}-\rho_*\right\|_2 = \left\|\rho_{Y^\Gamma}-\rho_*\right\|_2\\
&= \sqrt{\frac{d^2-1}{d^2}}\sqrt{1-s}.
\end{align}
Hence, the image of $\Xi$ is not only equi-entropic but also its three reduced states are equidistant from the maximally mixed state in the HS norm. The corresponding trace-norm deviation satisfies
\begin{equation}\label{eq:trace-HS-bound}
\left\|\rho_X-\rho_*\right\|_1 \leq d\left\|\rho_X-\rho_*\right\|_2 = \sqrt{d^2-1}\sqrt{1-S(X)}.
\end{equation}
So, for fixed $d$, any sequence satisfying $S(X)\to 1$ also satisfies $\rho_X\to\rho_*$ in both HS and trace norm. However, when $d$ also tends to infinity, the bound guarantees trace-norm convergence under the stronger sufficient condition $d^2(1-S(X))\to 0$. In particular, an entropy deficit of order $O(d^{-2})$ guarantees the HS defect of order $O(d^{-1})$ but the bound~\eqref{eq:trace-HS-bound} alone does not establish trace-norm convergence. Moreover, a small marginal AME defect does not by imply the existence of a nearby two-unitary matrix which is exact.

\section{Summary}\label{sec:summary}

We introduced and analyzed a linear map $\Xi$ that enforces equality of the normalized linear entropies associated with the three balanced bipartitions of a four-party system. 

The image of $\Xi$ consists of matrices with equal entropies, however, their common entropy need not be maximal. Special fixed points still can be two-unitary thus corresponding to exact AME$(4,d)$ states. For permutation matrices, the fixed-point condition imposes an additional cyclic symmetry on the associated orthogonal Latin squares. Numerical results reveal such symmetric two-unitary fixed points in several dimensions which motivate Conjecture~\ref{LS-conjecture}.

Acting on Haar-random unitary matrices, the map $\Xi$ produces highly entangled states for which the common entropy approaches the maximal value as the local dimension $d$ tends to infinity. Although $\Xi$ does not generally produce exact AME$(4,d)$ states, it provides a simple mechanism for constructing states with exactly equal entropies. In the leading-moment approximation, the average entropy deficit decreases as $O(d^{-2})$. Thus, the map enforces exact entropy equality in every dimension, while typical high-dimensional inputs make the common value close to maximal. The symmetry enforced by $\Xi$ accounts for the equality of the three entropies, while their near-maximal value reflects the statistics of high-dimensional random inputs.

It remains an open question whether the tensor rearrangements defining $\Xi$ admit a useful operational realization, possibly through an embedding into an enlarged Hilbert space. It would also be interesting to identify quantum-information tasks in which the equi-entropic states generated by $\Xi$ offer an advantage over generic highly entangled states. Exact AME states are known to be optimal resources for several quantum-information tasks~\cite{Casas_2026}, while approximate forms of multipartite maximal entanglement may remain useful under realistic, noisy conditions~\cite{approximate-k-uniform-states}. Since typical outputs of $\Xi$ have exactly equal bipartition entropies and an average entropy deficit that decreases with the local dimension, they provide natural candidates for investigation in this context. Establishing their operational usefulness, however, requires a notion of approximation beyond the linear-entropy deficit and is left for future work.

A continuous-variable (CV) extension provides another natural direction. At the level of tensor amplitudes, the index rearrangements underlying $\Xi$ can be formulated formally for four identical bosonic modes, but the finite-dimensional normalization of the linear entropy and the maximally mixed reference state do not extend directly to the CV setting. A physically meaningful formulation could impose an energy constraint or finite squeezing and characterize the balanced reductions through their purities or through covariance matrices and symplectic spectra, when considered in the Gaussian setting. This would connect the present construction with earlier work on Gaussian multipartite entanglement and recent formulations of CV AME states~\cite{FacchiGaussian2009,KBA25}.

\section{Acknowledgements}

We acknowledge discussions with Remigiusz Augusiak, Micha{\l} Eckstein, Mario Flory, Pawe{\l} Horodecki and {\L}ukasz Skowronek. Some of the ideas presented in this paper were first introduced in the PhD thesis of the first author and are further developed here. WB is supported by the National Science Centre (Poland) through the SONATA BIS project No. 2019/34/E/ST2/00369. The Center for Quantum-Enabled Computing project is carried out within the International Research Agendas programme of the Foundation for Polish Science co-financed by the European Union under the European Funds for Smart Economy 2021--2027 (FENG). 
ZR. acknowledges support from the German Federal Ministry of Research, Technology and Space (BMFTR) through the Quantum Futur project "Quantum Secure Communication with Adaptive Layers and Ranks (Q-SCALAR)", Grant No. 13N17638.

\bibliography{bibliography}

\appendix

\section{Explicit MOLS Associated with Symmetric Two-Unitary Permutations}\label{app:LS}

In this appendix, we provide complete sets of mutually orthogonal Latin squares found in the symmetry-constrained numerical search described in Sec.~\ref{sec:equi-entropic-map}. The two collections below form complete sets of $d-1$ mutually orthogonal LS for $d=7$ and $d=9$, respectively. Each LS is represented row by row.  The symbols are taken from $\{0,\ldots,d-1\}$.

\subsection{\texorpdfstring{MOLS$(7)$}{}}

\begin{widetext}
\begin{center}
{\small
\noindent\verb"2430516  3502146  2146035  4105263  6324105  5421360"
\verb"6324105  5421360  3502146  2430516  4105263  2146035"
\verb"3041652  1360254  6035421  1652430  0516324  0254613"
\verb"5263041  4613502  4613502  5263041  5263041  4613502"
\verb"4105263  2146035  5421360  6324105  2430516  3502146"
\verb"1652430  6035421  0254613  0516324  3041652  1360254"
\verb"0516324  0254613  1360254  3041652  1652430  6035421"}
\end{center}
\end{widetext}

\subsection{\texorpdfstring{MOLS$(9)$}{}}

\begin{widetext}
\begin{center}
{\small
\noindent\verb"135076284  470821356  746051832  761480253  164783052  831257046  178264503  567218430"
\verb"612453708  712534068  862105374  340275618  487150263  504361278  201873465  306182547"
\verb"254607813  567218430  583426017  457163082  273015648  780613524  746051832  712534068"
\verb"367821540  143607285  324510786  825631407  751462380  627084315  862105374  470821356"
\verb"728345061  681345702  037642158  213048576  310246875  152840637  037642158  681345702"
\verb"843760125  306182547  201873465  184752360  805624731  375408162  415387620  235760814"
\verb"570182436  854076123  178264503  608527134  042378516  463572801  650738241  028453671"
\verb"081534672  028453671  415387620  072316845  526837104  016725483  324510786  143607285"
\verb"406218357  235760814  650738241  536804721  638501427  248136750  583426017  854076123"}
\end{center}
\end{widetext}

\end{document}